\documentclass[12pt]{article}
\usepackage{hyperref}
\usepackage{amsmath, amsthm, amssymb, amscd, graphicx, dsfont, algorithmic}
\usepackage{natbib}
\bibpunct{(}{)}{;}{a}{}{,}
\theoremstyle{plain}
\newtheorem{theorem}{Theorem}
\newtheorem{lemma}{Lemma}

\newcommand{\real}{\mathbf{R}}
\newcommand{\N}{\mathbb{N}}
\renewcommand{\P}{\mathbf{P}}
\newcommand{\eps}{\varepsilon}
\DeclareMathOperator*{\argmin}{argmin}
\DeclareMathOperator*{\argmax}{argmax}
\DeclareMathOperator{\diverge}{div}
\DeclareMathOperator{\median}{Med}
\DeclareMathOperator{\Cov}{Cov}

\usepackage{ifthen}
\newboolean{blind}
\setboolean{blind}{false}

\title{Locally adaptive image denoising\\ by a statistical multiresolution criterion}
\ifthenelse{\boolean{blind}}{}{
\author{Thomas Hotz, Philipp Marnitz, Rahel Stichtenoth, \\ Laurie Davies, Zakhar Kabluchko and Axel Munk}
}
\date{November 2009}

\begin{document}

\maketitle

\begin{abstract}
We demonstrate how one can choose the smoothing parameter in image denoising by a statistical multiresolution criterion, both globally and locally. Using inhomogeneous diffusion and total variation regularization as examples for localized regularization schemes, we present an efficient method for locally adaptive image denoising. As expected, the smoothing parameter serves as an edge detector in this framework. Numerical examples illustrate the usefulness of our approach. We also present an application in confocal microscopy.
\\[6pt]\noindent
\\[6pt]\noindent
\end{abstract}

\section{Introduction}
\label{intro}

Image denoising is one of the main tasks in image analysis, as documented by numerous articles and books published on the subject, cf. e.g. \citep{ScheGraGroHaLe:2009}, \citep{BuaCoMo:2005} and \citep{AuKo:2002}.
Accordingly, statisticians have contributed their share: using probabilistic models for the true image, Bayesian methods were among the first to add a statistical perspective to the subject, see e.g. \citep{GeGe:1984}, \citep{Be:1986}, and \citep{Wi:2003}.
Taking a frequentist's point of view, image denoising becomes a smoothing or reconstruction problem, see e.g. \citep{HaTi:1986} and \citep{PoSpo:2000, PoSpo:2003} as well as \citep{KoTsy:1993}. Acknowledging the non-smooth nature of many images caused by sharp edges, cf. \citep{ChuGlaGoMa:1998} and \citep{Do:1999}, prompted a generalization of the well-established smoothing techniques, e.g. drawing on methods from one-dimensional jump detection, see \citep{Qiu:2005, Qiu:2007}; cf. also \citep{MraWeiBru:2006} for a unifying framework for many popular numerical and statistical denoising techniques.

Put simply, \emph{statistical image denoising} amounts to reconstructing a noiseless image $f$ given a noisy image $y$. Usually, one assumes the noise $\epsilon$ to be additive, i.e. $y = f + \epsilon$. In this article, we assume that the noise is generated at random; more specifically, we assume for pixels $(i,j)$ that
\begin{equation}
\label{model}
y_{ij} = f_{ij} + \epsilon_{ij}
\end{equation}
with $\epsilon_{ij}$ identically and independently distributed Gaussian random variables with zero mean and variance $\sigma^2$, i.e.
\begin{equation}
\label{noise}
\epsilon_{ij} \stackrel{\text{i.i.d}}{\sim} \mathcal{N}(0, \sigma^2);
\end{equation}
in particular we assume the value at each pixel to be a real number, i.e. $y_{ij}$, $f_{ij}$, $\epsilon_{ij} \in \real$. This models a grey-scale image, though in practice its grey levels are usually restricted to a finite number of discrete values, e.g. to integers between $0$ and $255$.
In many applications, a Gaussian assumption on the noise is therefore not very plausible but other noise processes will be more suitable. Nonetheless, for simplicity's sake we will lay out our basic ideas under a Gaussian assumption; possible extensions beyond this are briefly discussed at the end of Section~\ref{sec_mrc}. We note, however, that for a well-calibrated image $y$ making good use of the range of (not too few) possible values, this assumption is not very crucial, as long as the errors are i.i.d., symmetric and feature no heavy tails: the true image $f$ will then take values well inside the range, itself not necessarily being restricted to discrete values -- a property it passes onto the errors $\epsilon_{ij}$.

Figure~\ref{figdata} shows an artificial example where $f$ exhibits varying degrees of smoothness (left), and Gaussian white noise has been added to obtain the data $y$ (right). The assumption about the noise can then be exploited in order to distinguish between the `true', noiseless image $f$ and the noise $\epsilon$ as demonstrated in the following sections.

\begin{figure}[!tb]
\begin{center}
\setlength{\tabcolsep}{2mm}
\begin{tabular}{clcl}
\includegraphics[width=150pt]{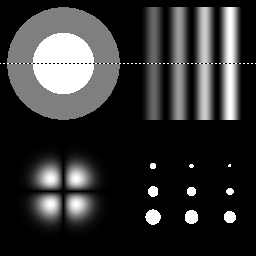}
&
\includegraphics[height=150pt]{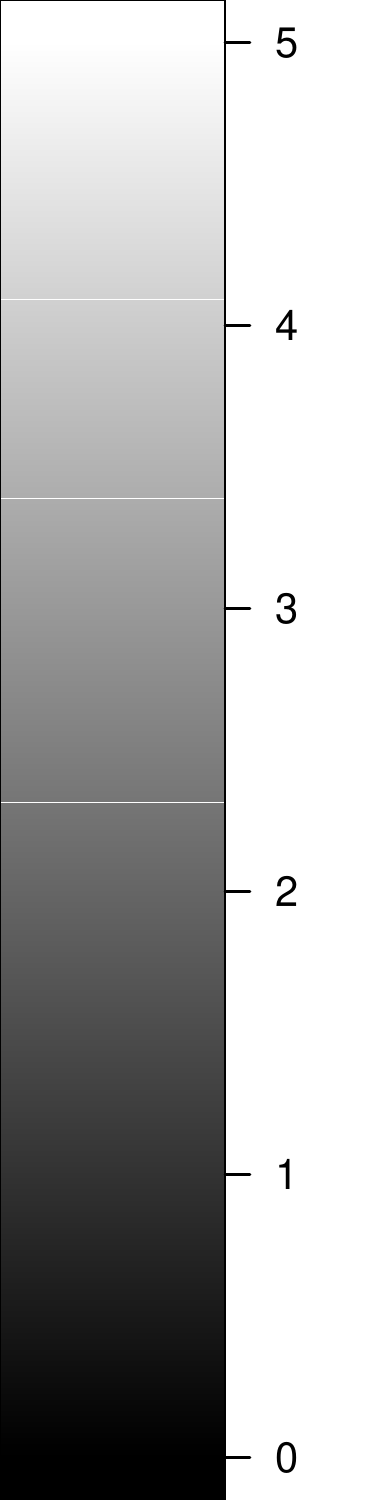}
&
\includegraphics[width=150pt]{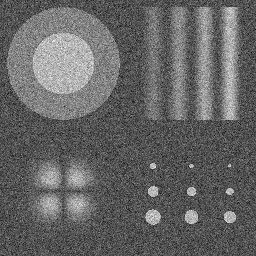}
&
\includegraphics[height=150pt]{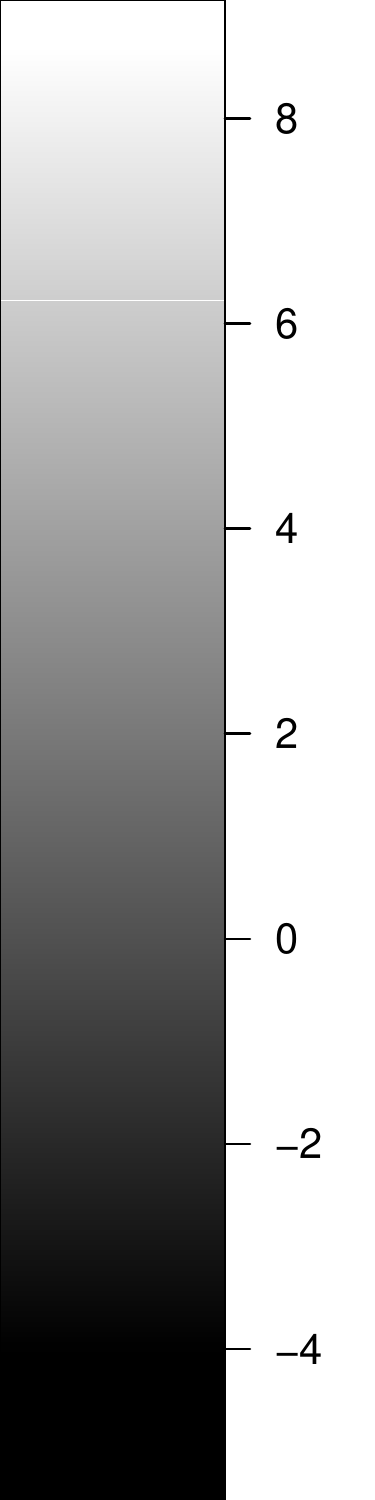}
\\
(a) & & (b) &
\vspace*{-12pt}
\end{tabular}
\end{center}
\caption{\label{figdata}(a) $256 \times 256$ pixel test image $f$ taking values in $[0,5]$; dashed line indicates row 64 shown in detail in Figure~\ref{cut}. (b) Simulated data $y$ with noise level $\sigma = 1$.}
\end{figure}

Clearly, this is impossible unless assumptions are made about $f$, e.g. that it varies slowly from pixel to pixel. In order to be able to formulate such `smoothness' assumptions more precisely, we let (in a slight abuse of notation) $f_{ij}$ denote the values of a function $f$ on a regular grid, i.e. $f:[0,1]^2 \rightarrow \real$ with $f_{ij} = f\big(\frac{i}{n}, \frac{j}{n}\big)$ where $i,j \in \{1, \dots, n\}$. Note that we use a square grid for ease of notation, with $x_{ij} = \big(\frac{i}{n}, \frac{j}{n}\big)$ in the following; one easily can extend all the analyses to rectangular grids, to higher dimensions, or to non-uniform sampling schemes through the use of finite elements, if required, cf. e.g. \citep{ErnGue:2004}. Now that $f$ can be viewed as a function, `smoothness' can be defined more rigorously to mean that $f$ belongs to some function class, e.g. that 
$f$ is in a Sobolev or Besov ball, or that $f$ has bounded variation, cf. e.g. \citep{KoTsy:1993}.

Image denoising techniques like the ones we discuss in Section~\ref{sec_denois} generally require the choice of a \emph{smoothing} or \emph{regularization parameter} $a$ which determines how much smoothing is to be applied. This parameter might be \emph{localized}, allowing for different amounts of smoothing to be applied to different parts of the image; the regularization parameter then becomes a function $a:[0,1]^2 \rightarrow \real$. Since one might want to smooth more where the true image $f$ is smoother, this then enables one to adapt to differing levels of smoothness across the image. The reconstruction or denoised image $\hat f$ hence depends on that smoothness parameter. The aim of this research was to find a purely data-driven and generally applicable way to choose~$a$. This will be illustrated with two specific denoising techniques, namely for linear diffusion as well as for TV penalization. We stress, however, that our approach is in principle applicable to any regularization technique which depends on properly choosing a regularization parameter, the latter possibly being a function as described above.

The main idea can be summarized as follows: consider the residuals $r_{ij} = y_{ij} - \hat f_{ij}$; they depend on the smoothness parameter $a$. Indeed, if we smooth too much some structures which were present in $f$ have been smoothed away -- these structures then will be left in the residuals. Had we found the perfect reconstruction, i.e. for $\hat f = f$, however, the residuals form white noise, see \eqref{model}. One possible way to decide whether we smoothed too much is therefore to check whether the residuals look like white noise -- if there is still some structure left in them we must have smoothed too much. We note that this idea is at the heart of statistical methods for automatically selecting the regularization parameter.

As the key ingredient for choosing the smoothness parameter $a$ we propose a \emph{statistical multiresolution criterion}. Introduced in Section~\ref{sec_mrc}, it measures deviations from the hypothesis that the residuals are white noise. An important feature of this criterion is that it not only detects if the residuals deviate from white noise but also \emph{where}. This then allows for a locally adaptive choice of $a$.  In Section~\ref{sec_adapt} we derive an algorithm for a data-driven selection of $a$. Numerical details and results are given in Section~\ref{sec_num}, alongside an example from confocal microscopy. Finally, we summarize and discuss what we have achieved in Section~\ref{sec_discuss}. The interested reader might also find the Ph.D. thesis of \cite{Sti:2007} useful, on which parts of this article have been based.

While there exists an extensive literature on localized, data-driven ways of choosing the smoothing parameter in one-dimensional function estimation, see e.g. \citep{LeMaSpo:1997}, \citep{GijMa:1998} and \citep{DueSpo:2001}, and several articles have been published on multiresolution criteria for one-dimensional settings, cf. the references in Section~\ref{sec_mrc}, higher-dimensional situations have scarcely been treated. In fact, to the authors' knowledge \citep{BiMaMu:2006, BiMaMu:2008a} and \citep{DaMei:2008} are the only published work on a data-driven choice of the smoothing parameter by statistical multiresolution techniques in higher dimensions; however, \cite{BiMaMu:2006, BiMaMu:2008a} essentially still apply a one-dimensional version of the multiresolution criterion to this end whereas the two-dimensional application in \citep{DaMei:2008} is described very briefly and somewhat rudimentary.

The present paper therefore appears to be the first to give a comprehensive and detailed exposition of a two-dimensional multiresolution criterion, and also in proposing a method for choosing a localized smoothing parameter in a purely data-driven and general fashion. We claim that it can be applied to a wide range of reconstruction techniques depending on a global or local regularization parameter; this will be demonstrated exemplarily for linear diffusion filtering and total variation regularization below.

Let us now consider these commonly used methods for image reconstruction, i.e. for noise removal.

\section{Image denoising}
\label{sec_denois}

Among the best studied denoising technique in image processing is the \emph{linear diffusion filter}, see e.g. \citep{Wei:1998} or \citep{GuiMoRya:2004}, which has the physical interpretation of an evolution of the \emph{heat equation}. For constant diffusivity $a \in \real$, it is given by
\begin{equation}
\label{homdiff}
\begin{cases}
\frac{\partial}{\partial t} u_a(x,t) = a\ \triangle_x u_a(x,t) \\
u_a(x, 0) = y
\end{cases}
\end{equation}
which is solved by the convolution of the initial data $y:\real^2 \rightarrow \real$ with a Gaussian kernel, i.e.
\begin{equation}
\label{solution}
u_a(x_*,t) = \int K_{\sqrt{2at}}(x - x_*) y(x) dx
\end{equation}
with
\begin{equation}
\label{kernel}
K_h(x) = \frac{1}{2 \pi h} \exp \Big( -\frac{\Vert x\Vert^2}{2h^2} \Big);
\end{equation}
it is therefore also called \emph{Gaussian filtering}.
This diffusion filter is very well understood by now: it is a low-pass filter effectively reducing white noise; the heat equation is the limit of repeated convolving with any kernel fulfilling some moment conditions \citep[Ch.~2]{GuiMoRya:2004}; it gives rise to a linear scale space with many desirable properties \citep{ChauMa:2000}; and there are many sets of such properties which uniquely determine this scale space, see e.g. \citep{Wi:1983}, \citep{Koe:1984}, \citep{AlGuiLioMo:1993}, or \citep[Table~1.1]{Wei:1998} for an overview.

Note that the aim of this paper is not to propagate this specific filter but rather to illustrate with this particular example how statistical multiresolution analysis can be used to effectively choose the diffusivity, or some other regularization parameter, globally or locally -- even in the case when the optimal choice depends on the smoothness of the true image $f$; other regularization techniques will be discussed at the end of this section.

In statistics, model \eqref{model}, where $f$ is not assumed to have a certain low-dimensional, parametric form but instead to belong to some non-parametric function class, is called a (two-dimensional) \emph{nonparametric regression model}. A standard approach to estimate $f$ from the data $y$ is to use a \emph{kernel estimator} $\hat f_K$: choosing some kernel $K:\real^2 \rightarrow \real$, one forms a local average where the kernel defines the weights,
\begin{equation}
\hat f_K(x_*) = \frac{\sum_{ij} K(x_{ij} - x_*) y_{ij}}{\sum_{ij} K(x_{ij} - x_*)},
\end{equation}
for arbitrary $x_* \in [0,1]^2$. A popular choice for $K$ is the isotropic Gaussian kernel $K_h$ with bandwidth $h$ given in \eqref{kernel}.
We refer to \citep{WaJo:1995} for a broader treatment of kernel estimation, mentioning just one of the many textbooks on the subject.

From \eqref{solution} we see that, after discretization in space and ignoring boundary effects, the solution of the heat equation at time $t$ is the kernel estimator with bandwidth $h = \sqrt{2at}$, i.e. the r\^ole the (global) bandwidth is playing in kernel estimation is played by time when smoothing with the heat equation. Alternatively, we can vary the diffusivity and fix the endpoint in time since $u_1(x, t) = u_{t}(x, 1)$. We will take the latter approach, calling the process \emph{homogeneous diffusion}, with the diffusivity $a \in \real$ to be specified. Finally, we also discretize time to only two time points, namely $0$ and $1$, thereby obtaining a one-step approximation to the heat equation which is computed on the same grid as the data are on. For $a \in \real$, this defines an estimator $\hat f_\text{hom.diff.}$ through
\begin{equation}
\label{homdiffest}
\hat f_\text{hom.diff.} - y = a\ \tilde\triangle_x \hat f_\text{hom.diff.}
\end{equation}
where $\tilde\triangle_x$ is a discretization of the Laplacian in space.

Clearly, the choice of the diffusivity -- or, equivalently, of the bandwidth -- is critical for the performance of the estimator: a large diffusivity will reduce the variance of the estimator but at the expense of increasing its bias. While the variance depends on the noise level $\sigma$, the smoothness of the function influences the bias: the smoother the function, the less the bias, with constant functions giving no bias at all.
This is illustrated in Figure~\ref{varyglobal}: while a large diffusivity (left) is permissive in regions where the underlying signal, see Figure~\ref{figdata}(a), is smooth, e.g. in the interior of the large circle or within the background, sharp edges and smaller features like the circles in the lower right or the ``valleys'' in the lower left are smoothed away; the latter need a lower diffusivity (right) but this then results in unnecessary and unwanted undersmoothing in the upper left part.
The optimal diffusivity strikes a balance between bias and variance -- a goal which is difficult to achieve, given that the bias depends on the smoothness of the \emph{unknown} function. An estimator that chooses the diffusivity in a purely data-driven manner will be called \emph{adaptive}.

\begin{figure}[!tb]
\begin{center}
\setlength{\tabcolsep}{2mm}
\begin{tabular}{ccl}
\includegraphics[width=150pt]{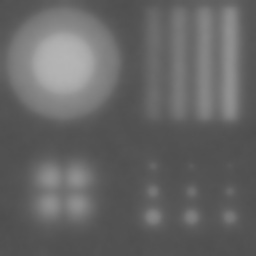}
&
\includegraphics[width=150pt]{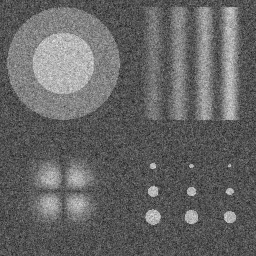}
&
\includegraphics[height=150pt]{sn5_pal.pdf}
\\
(a) & (b) &
\vspace*{-12pt}
\end{tabular}
\end{center}
\caption{\label{varyglobal}Homogeneous diffusion applied to the data in Figure~\ref{figdata}(b), with diffusivity chosen as $a = 40$ in (a), and $a = 0.1$ in (b).}
\end{figure}

Moreover, instead of choosing a global, one-fits-all, diffusivity $a$, one would rather want to choose it locally, as the function's smoothness is a local feature, too, cf. again Figure~\ref{varyglobal}. Estimators that choose the diffusivity locally in dependence on the data observed are \emph{locally adaptive}. We stress that in the context of our work ``adaptivity'' is not meant in the sense that the estimator attains certain optimality properties over scales of spaces (see e.g. \citep{Le:1990}, \citep{Tsy:1998} and \citep{DueSpo:2001}), but we rather use this term heuristically to express that the estimator chooses e.g. its diffusivity in a data-driven fashion, possibly taking local smoothers of $f$ into account. Deriving such an estimator is the aim of this paper. It will be achieved by the use of the \emph{statistical multiresolution criterion} which allows to statistically decide whether some reconstruction's residuals still contain parts of the signal. We defer to Section~\ref{sec_mrc} for a longer discussion of its analytical properties and to Section~\ref{sec_adapt} for its application to image reconstruction.

The task of choosing the bandwidth locally is similar to selecting a spatially varying diffusivity $a(x)$ in \eqref{homdiff}, and correspondingly in \eqref{homdiffest}. We point out that, nonetheless, these techniques are not equivalent; this can easily be seen if one compares the effect of a zero diffusivity along a curve which inhibits the exchange of information across even if the diffusivity is large close to that curve. Hence, local diffusivities allow to respect sharp edges. Note that in the case of a local diffusivity $a:[0,1]^2 \rightarrow \real$, equation \eqref{homdiffest} can still be solved very efficiently as opposed to the computation of kernel estimators with varying bandwidths, see Section~\ref{sec_num} for details. We thus obtain an estimator $\hat f_\text{inhom.diff.}$ through this \emph{inhomogeneous diffusion} process. The corresponding linear operator will be denoted by $L_a$, i.e.
\begin{equation}
\label{inhomdiffest}
\hat f_\text{inhom.diff.} = L_a y = a\ \tilde\triangle_x L_a y + y.
\end{equation}
For details on how to solve this equation as well as on differences to local bandwidth selection, we refer to \citep{Sti:2007}.

The variational formulation of the continuous version of \eqref{inhomdiffest} is given by
\begin{equation}
\label{varinhomdiff}
\argmin_{g \in H^2}\ \Big\Vert \frac{ g - y }{a} \Big\Vert^2 + \int \vert \nabla g \vert^2
\end{equation}
where $\Vert \cdot \Vert$ and $\vert \cdot \vert$ denote the $L_2$-norm and the Euclidean norm in $\real^2$, respectively. Here, $H^2$ denotes the Sobolev space of functions in $L_2$ having weak second derivatives also in $L_2$, i.e. $H^2 = \{ g\ :\ \sum_{\vert \alpha \vert \leq 2} \Vert D^\alpha g \Vert < \infty \}$. Clearly, the first term in \eqref{varinhomdiff} is a weighted data fit whereas the second term is a smoothness penalty; this allows to view $a$ also as local weights, or being proportional to a locally varying standard deviation $\sigma$ of the errors.

Other penalties can also be thought of: considering the poor performance of the kernel estimator at sharp edges, one might also be interested in a total variation (TV) penalty as introduced by \cite{RuOshFa:1992},
\begin{equation}
\label{vartv}
\hat f_\text{TV} = \argmin_{g \in BV}\ \Big\Vert \frac{ g - y }{a} \Big\Vert^2 + \int \vert \nabla g \vert\, ,
\end{equation}
where the second term symbolically stands for the TV semi-norm, i.e. $\int \vert \nabla g \vert := \sup_{\phi \in C_c^1(\real^n), \Vert \phi \Vert_\infty \leq 1} \int g \diverge \phi$ where $\phi \in C_c^1(\real^n)$ iff $\phi \in C^1(\real^n)$ and $\phi$ has compact support. We also assume a discretization in space as above; as usual, $BV$ denotes the space of functions with bounded variation, defined as $BV = \{ g\in L_1\ :\ \int \vert \nabla g \vert < \infty \}$.

The main question that needs to be answered is how to choose the diffusivity (or weights) $a$. This will be addressed in the following two sections. We note that it is not always possible to localize the regularization parameter. One class of denoising techniques for which this is the case are iterative methods; there, the iteration at which to stop the algorithm plays the r\^ole of the regularization parameter, see e.g. \citep{BiMaMu:2006, BiMaMu:2008a} for the EM algorithm, a.k.a. Richardson-Lucy algorithm, or \citep{BiHoRuyMu:2007} for a general treatment of iterative regularization schemes. We therefore also discuss how to select $a$ globally, although the focus will be on a local selection strategy.

\section{A statistical multiresolution criterion}
\label{sec_mrc}

We have assumed the errors $\epsilon_{ij}$ to be white noise, see \eqref{noise}. Clearly, considering the residuals $r_{ij} = y_{ij} - \hat f_{ij}$ for some estimator $\hat f$, we would want them to be as close to white noise as possible -- if there is any structure left in the residuals then the estimator must have missed essential features of the true $f$. In this section, we therefore aim to derive a statistical test for the hypothesis that the residuals are white noise, i.e. $r_{ij}  \stackrel{\text{i.i.d}}{\sim} \mathcal{N}(0, \sigma^2)$, against the alternative that there is some structure left in them. To this end we define random variables
\begin{equation}
\omega_P = \frac{1}{\sqrt{\sharp\{ x_{ij} \in P\}}} \sum_{x_{ij} \in P} r_{ij}
\end{equation}
for suitable subsets $P \subseteq [0,1]^2$ of the image domain, to be specified later. We call $\omega_P$ the \emph{multi\-resolution coefficient} of $P$.
The collection ${\mathcal P}_n$ of
subsets $P$ cannot be too large and will typically be of polynomial
order $O(n^k)$ of the number of observations $n$. This is the case if
${\mathcal P}_n$ is a Vapnik-Cervonenkis (VC) class
of subsets, cf. \citep{Po:1984,DeLu:2001}. In fact the family $\mathcal P_n$ must be a VC-class for
the procedure to make sense as otherwise too many subsets are picked out and the stochastic fluctuation of the random variables $\omega_P$ can no longer be controlled simultaneously over all subsets $P \in \mathcal P_n$.
The choice of ${\mathcal P}_n$ is subtle as it influences
the limit behaviour; it will be addressed later where we propose
amongst others a dyadic squares partitioning.

Note that $\omega_P \sim \mathcal{N}(0,\sigma^2)$ under the hypothesis. If however there is some information left in the residuals then the mean of some residuals will no longer be $0$. Hence the corresponding multiresolution coefficients $\omega_P$ are also no longer distributed around $0$ but around the mean of the signal left in the residuals, thus becoming large in absolute value in comparison to the expected behaviour of white noise.

Given an appropriate collection of subsets $\mathcal P_n$ of $[0,1]^2$
and $\alpha, 0 < \alpha <1,$ we can determine (e.g. through simulations) some
$\tau_n=\tau_n(\alpha)$ such that
\begin{equation}
\P\left( \max_{P \in \mathcal{P}_n}
\frac{\big\vert \sum_{x_{ij} \in P} Z_{ij} \big\vert} {\sqrt{\sharp\{
    x_{ij} \in P\}}}\le \sqrt{\tau_n\log n^2}\right)=\alpha, 
\end{equation}
where $Z_{ij}$ are standard Gaussian white noise random variables. For
any function $g:[0,1]^2\rightarrow \real$ and $P \in \mathcal P_n$ we define
\begin{equation} \label{sumresid}
w(g,P)=\frac{\sum_{x_{ij} \in P}
  (y_{ij}-g_{ij})}{\sqrt{\sharp\{ x_{ij} \in P\}}} 
\end{equation}
where $g_{ij}=g(i/n,j/n)$ and put
\begin{equation} \label{approxreg}
{\mathcal F}_n={\mathcal F}_n({\mathcal P}_n,\sigma,\tau_n)= \big\{g:
\max_{ P \in {\mathcal P}_n} \vert w(g,P)\vert \le \sigma \sqrt{\tau_n
  \log n^2\,}\,\big\}.
\end{equation}
For data generated under (\ref{model}) with $\epsilon_{ij}=\sigma
Z_{ij}$ it is seen that ${\mathcal F}_n$ is an exact, universal and
non-asymptotic confidence region for any $f: [0,1]^2 \rightarrow\real$:
\[\P(f\in {\mathcal F}_n)=\alpha\]
provided this event is measurable, see \citep{DaKoMei:2009}. The confidence region $\mathcal F_n$ contains many functions which
are of little interest, for example, all functions which interpolate
the data. This may be seen as a severe case of overfitting. In general,
however, we are interested in simple or, if possible, the simplest
functions in $\mathcal F_n$ where simplicity may be defined in terms of
smoothness, shape or sparsity or some combination of all
three. Regularization within $\mathcal F_n$ leads to optimization
problems such as
\begin{equation}
\label{eq:tvmin}
\text{minimize}\quad TV(g^{(k)})\quad\text{subject to}\quad g \in {\mathcal
F}_n\quad\text{for some}\quad k=0,1,\ldots
\end{equation}
where $TV(g^{(k)})$ denotes some definition of total variation of the
function $g^{(k)}$. For examples of this approach for one-dimensional data
we refer to \citep{MaVDGee:1997} and \citep{DaKoMei:2009}.

 In some cases where the optimization problem is algorithmically too
 difficult to be solved we may nevertheless have a sequence ${\tilde
   f}_m,m=1,2,\ldots$ of good candidate functions of increasing 
complexity. This is the case we consider in this paper and the strategy
consists of choosing the first function ${\tilde f}_m$ which lies in
$\mathcal F_n$. The success of this strategy depends largely on how good
the candidate functions are. 

The definition of ${\mathcal F}_n$ involves $\sigma$ which has to be estimated from the data. We propose a robust estimator for this purpose in Section~\ref{sec_adapt}, see (\ref{sigmahat}). We note that it would be possible to refine the simple substitution of ${\hat \sigma}_n$ for $\sigma$ slightly so that  ${\mathcal F}_n$ becomes an honest \citep{Li:1989, GeWa:2008}, universal and non-asymptotic confidence region for $f$, i.e.
\[\P(f\in {\mathcal F}_n)\ge \alpha.\] 

As already mentioned $\tau_n(\alpha)$ can be obtained by
simulations. However, if the asymptotic behaviour of $\tau_n$ can be
determined then it is often sufficient to use $\tau_{\infty}=\lim_{n
  \rightarrow \infty}\tau_n$. This is not a simple problem and we
consider it in more detail below. In the cases we consider we have
$\tau_{\infty}=2$.

Following these considerations, we define the
\emph{statistical multiresolution criterion} $M$ by
\begin{equation}
\label{mrc}
M_n = \frac{1}{\sqrt{2 \log n^2}} \max_{P \in \mathcal{P}_n} \vert \omega_P \vert = \frac{1}{\sqrt{2 \log n^2}} \max_{P \in \mathcal{P}_n} \frac{\big\vert \sum_{x_{ij} \in P} r_{ij} \big\vert} {\sqrt{\sharp\{ x_{ij} \in P\}}},
\end{equation}
for some collection $\mathcal{P}_n$ of subsets $P$, the normalization factor $\sqrt{2 \log n^2}$, or $\sqrt{2 \log n^d}$ in a $d$-dimensional setting, becoming clear in a short while.
Note that the choice of $\mathcal{P}_n$ is subtle as it influences the limit behaviour; it will be addressed later where we propose to use a dyadic squares partitioning.

The multiresolution criterion in the form of \eqref{mrc} was introduced by \cite{SieVe:1995} to detect change points; \cite{DaKo:2001} were first to use it for one-dimensional non-parametric regression, \cite{BiMaMu:2006, BiMaMu:2008a} applied it to positron emission tomography, while a similar criterion has been introduced by \cite{DueSpo:2001} as well as \cite{DueWa:2008} in the context of testing qualitative hypotheses in non-parametric regression. A major difference between \eqref{mrc} and the latter authors' multiscale statistic is that they calibrate the average residuals by a term of the order $\sqrt{\log(n / \sharp P})$ in order to enhance medium and large scales. In many imaging problems, however, the features on the small scales are most important as they reflect rapid local change at edges. This particularly results in a completely different limiting behaviour, compare Theorem~2 in \citep{DueWa:2008} with Theorem~\ref{thm_gen} below. \cite{SieWo:1995} test for the presence of a signal of unknown scale and position in arbitrary dimensions, using Gaussian weights in their multiresolution criterion; they derive its asymptotic distribution and discuss its power. \cite{DaMei:2008} used this criterion in one- and two-dimenstional settings to determine the weights of smoothing splines.

Although the $\omega_P$ are identically distributed they are dependent, rendering the distribution of $M_n$ difficult to obtain analytically. Recently, \cite{KaMu:2008b} established a.s. convergence, i.e.
\begin{equation}
\label{shao}
M_n \rightarrow \sigma \text{ a.s. for } n \rightarrow \infty
\end{equation}
for $\mathcal{P}_n$ the collection of all squares and rectangles; for the one-dimensional case and intervals this was shown by \cite{Shao:1995}. For the latter case, \cite{SieVe:1995} proved that $M_n$ is asymptotically Gumbel-distributed, cf. also \citep{Ka:2007}; \citep{SieYa:2000} suggests this also to hold in the multi-dimensional case. We will present a similar result for the collection of dyadic cubes in Theorem~\ref{thm_dyadic} on page~\pageref{thm_dyadic}. It will be based on the following, general theorem which is proved in the appendix:
\begin{theorem}
\label{thm_gen}
For each $N\in\N$, let $(\xi_1^{(N)},\ldots,\xi_N^{(N)})$ be a Gaussian vector with standardized marginal distributions. Suppose that for every $\eps>0$ and some constant $\rho<1$ not depending on $N$ we have
\begin{equation}\label{eq:sparse}
\#\{(i,j)\in\{1,\ldots,N\}^2: \Cov(\xi_{i}^{(N)},\xi_{j}^{(N)})\neq 0 \}=O(N^{1+\eps})\text{ as } N\to\infty
\end{equation}
and
\begin{equation}\label{eq:bounded_cov}
|\Cov(\xi_i^{(N)},\xi_j^{(N)})|\leq\rho \text{ provided that } i\neq
j.
\end{equation}
Then
\begin{equation}\label{eq:stat_theo}
\lim_{N\to\infty}\P\left[\max_{i=1,\ldots,N}\xi_i^{(N)}\leq a_N+b_N\tau\right]=\exp(-e^{-\tau})
\end{equation}
where $a_N$ and $b_N$ are sequences of constants defined by
\begin{equation}\label{eq:defab}
a_N=\sqrt{2\log N}+\frac{-1/2 \log\log N -\log 2\sqrt{\pi}  }{\sqrt{2\log N}},\qquad b_N=\frac{1}{\sqrt{2\log N}}.
\end{equation}
\end{theorem}
\noindent
We note that condition~\eqref{eq:sparse} states that the covariance matrix of $(\xi_i^{(N)})_{i=1,\ldots,N}$ is ``sparse'', that is, it contains at most $O(N^{1+\eps})$ non-zero elements. Condition~\eqref{eq:bounded_cov} states that the covariance matrix has mutual coherence less than $1$, i.e. its off-diagonal elements are bounded away from $\pm 1$.

While the exact distribution of $M_n$ might not be available for more general collections $\mathcal{P}_n$ or in a finite setting, a critical value for testing the hypothesis can in principle be obtained through simulation.

We note that extensions to more general error distributions are possible, see \citep{Shao:1995}, \citep{KaMu:2008b} or \citep{NaSieYa:2008}. For the particular case of Poisson data with not too small intensities, we suggest to approximate these by Gaussian distributions: hypothesising that the data $y$ stem from intensities $\hat f$, compute residuals $r_{ij} = \hat f_{ij}^{-1/2} ( y_{ij} - \hat f_{ij} )$ which are approximately i.i.d. like Gaussian white noise with standard deviation $1$, see also Section~\ref{sec_num}.

\section{Data-driven choice of the diffusivity}
\label{sec_adapt}

We now return to the question of how to choose the diffusivities for the estimators defined in Section~\ref{sec_denois}. In view of the test criterion in \eqref{mrc}, we require that the conditions
\begin{equation}
\label{mrcineq}
\vert \omega_P \vert \leq \sigma \sqrt{\delta \log n^2}
\end{equation}
hold for all $P \in \mathcal{P}_n$. The asymptotics in \eqref{shao} suggest $\delta > 0$ to be chosen close to $2$, in such a way that the error of the first kind is below a certain, prespecified significance level $\alpha$, say $\alpha = 5\%$. If all these conditions are fulfilled, we cannot reject the hypothesis that the residuals are white noise, or equivalently that the reconstruction $\hat f$ is the true function $f$; thence we are to accept $\hat f$ as possibly being the true $f$.

An important property of the inequalities \eqref{mrcineq} to note is that they are only one-sided, i.e. they are only violated if the multiresolution coefficients $\omega_P$ get large in absolute value. This will happen if we oversmooth, using too large a diffusivity. If we smooth too little, using too small a diffusivity, however, then the residuals get too small and so do the $\omega_P$, and the criterion will not be violated. For example, estimating $f$ by $\hat{f} = y$ leads to all residuals $r_{ij}$ being $0$, and hence the multiresolution criterion is trivially fulfilled.

Nonetheless, in many situations the collection $\mathcal{P}_n$ together with the candidate sets $\mathcal{F}_n$ constitute a nested sequence of increasing complexity when $\tau_n$ increases in \eqref{approxreg}; then it is reasonable to choose $\hat f_n \in \mathcal{F}_n$ such that the maximum of the left hand side of \eqref{mrcineq} is as close as possible to equality. \cite{BiMaMu:2006,BiMaMu:2008a} have investigated this strategy for stopping the EM algorithm in positron emission tomography. The reasoning behind is that equality is in fact obtained asymptotically as $n \rightarrow \infty$, cf. \eqref{shao} and Theorem~\ref{thm_dyadic}. Alternatively, one can choose the `simplest' function according to some complexity criterion, e.g. the TV-norm as in \eqref{eq:tvmin}.

\paragraph{Global choice.}
In light of the last remark, our strategy is to determine the largest diffusivity -- or the smoothest solution -- such that the inequalities \eqref{mrcineq} hold. For the homogeneous diffusion process \eqref{homdiffest}, this can be achieved algorithmically by starting with a large diffusivity $a \in \real^+$ which leads to oversmoothing, i.e. the corresponding estimator $\hat f_\text{hom.diff.}$ leads to residuals $r_{ij}$ that violate the multiresolution criterion \eqref{mrcineq}. Then, the diffusivity is reduced by a prespecified amount until the corresponding estimator gives residuals that fulfil the criterion. We refer to Section~\ref{sec_num} for more details.

\cite{DaKo:2001} have shown that such a strategy leads to consistent estimators when combined with the one-dimensional taut string, i.e. when the global smoothing parameter is given by the number of extreme values of the taut string estimator $\hat f:\real \rightarrow \real$; an application to the selection of peaks in X-ray diffractograms is given in \citep{DaGaMeiMeMi:2008}. \cite{BoLieMuWi:2007, BoKeMuLieWi:2009} proved a similar consistency result for selecting the number of jumps in jump regression. 

\paragraph{Dyadic squares partitioning.}
We yet have to specify what collection $\mathcal{P}_n$ of subsets we will use; for the sake of brevity such collections will be called \emph{partitionings}. Obviously, they do not constitute what is commonly referred to as a \emph{partition} of a set, i.e. a disjoint composition of the latter, but they rather comprise several partitions at different scales as we shall see. Indeed, we want such a partitioning to consist of subsets that allow to detect deviations from the hypothesis at different resolutions, both at coarse and fine scales. However, it should not be chosen too rich either. From a statistical point of view, taking too many subsets results in too many multiple tests being performed which can be seen analytically from the asymptotics in \eqref{shao} breaking down. Also, for each subset $P \in \mathcal{P}_n$ its multiresolution coefficient $\omega_P$ needs to be determined in our iterative procedure, rendering a large partitioning $\mathcal{P}_n$ computationally very expensive.
Furthermore it is also necessary that the condition for each individual $P \in {\mathcal P}_n$ can be quickly checked. We now show how all these conditions can be met.

The subsequent theory is based on a \emph{dyadic squares partitioning}, whereas for most practical purposes we use more subtle partitionings such as wedgelets or curvelets, see below. The dyadic squares partitioning originates from \citep{Do:1997}, cf. also \citep{KoJuGo:2005} and \citep{AnBiSa:2009} for applications to image segmentation and classification, resp. It is obtained by splitting the image recursively into four equal subsquares until some pre-specified lowest scale is reached, cf. Figure~\ref{figdyad} (left). This partitioning covers a wide range of scales with comparatively few subsets. At the same time it allows fast computation of the multiresolution coefficients through cumulative sums: define the matrix $R$ of cumulative sums by
\begin{equation}
 R = \left( \sum_{k=1}^{i} \sum_{l=1}^{j} {r_{kl}} \right)_{ij}.
\end{equation}
From $R$, the multiresolution coefficient $\omega_P$ of a rectangle $P = [i_1,i_2] \times [j_1,j_2]$ can readily be obtained by
\begin{multline}
\label{rect}
\sum_{{(i,j)} \in P} {r(x_{ij})} = R_{i_2,j_2} - R_{i_1-1,j_2} \mathds{1}_{\left\{ i_1 > 1 \right\}} - R_{i_2, j_1-1} \mathds{1}_{\left\{ j_1 > 1 \right\}} \\+ R_{i_1-1,j_1-1} \mathds{1}_{ \left\{ i1 > 1, j1 > 1 \right\} }.
\end{multline}

\begin{figure}[!tb]
\begin{center}
\includegraphics[width = 10 cm, height = 5cm]{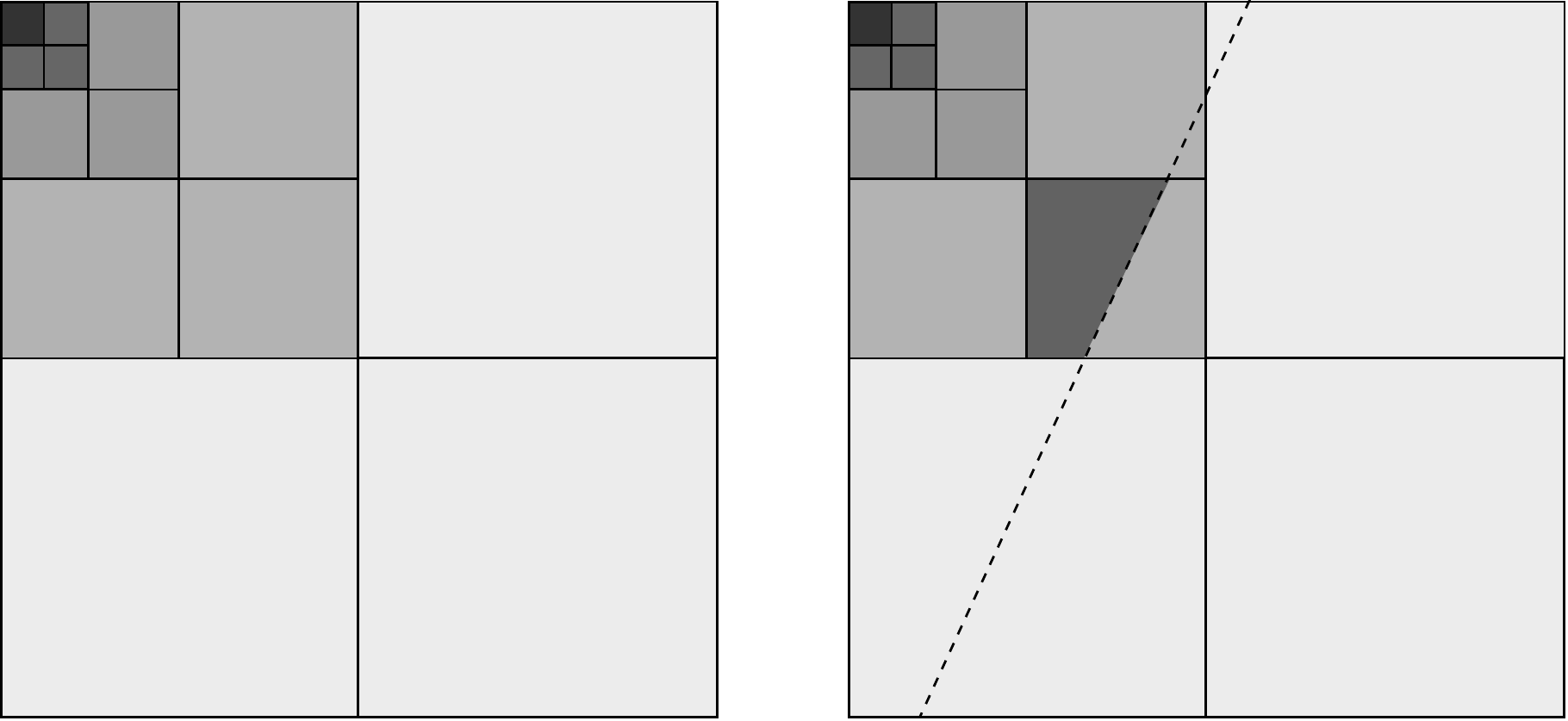}
\end{center}
\caption{\label{figdyad}Different scales of a dyadic squares partitioning (left), and with a line cutting through thereby creating wedgelets (right).}
\end{figure}

As promised in Section~\ref{sec_mrc}, we will now give the asymptotic distribution of $M_n$ in~\eqref{mrc} for a dyadic partitioning of a cube $K_n=\{1,\ldots,n\}^{d}$ in arbitrary dimension~$d$. First, let us recall the following well-known fact, see e.g.~\cite[Theorem 1.5.3]{LeLiRoo:1983}. If $\{\eta_i,i\in\N\}$ are independent standard Gaussian random variables, then for every $\tau\in\real$
\begin{equation}
\label{eq:equatmain}
\lim_{n\to\infty}\P\left [\max_{i=1,\ldots,n}\eta_i\leq a_n+b_n\tau \right ]=\exp(-e^{-\tau}),
\end{equation}
where the normalising sequences $a_n$, $b_n$ are defined in \eqref{eq:defab}.
As it turns out, we obtain the same asymptotic distribution, a \emph{Gumbel} distribution in the case of a dyadic partitioning:
\begin{theorem}
\label{thm_dyadic}
If $\mathcal{P}_n$ is the collection of dyadic subcubes of $K_n$ and $(r_i)_{i \in K_n}$ are independent standard Gaussian random variables, then for every $\tau\in\real$,
\begin{equation}
\label{eq:fish_tip}
\lim_{n\to\infty} \P[M_n\leq a_{\sharp\mathcal{P}_n}+b_{\sharp\mathcal{P}_n}\tau]=\exp(-e^{-\tau}),
\end{equation}
where
\begin{equation*}
M_n = \max_{P \in \mathcal{P}_n} \frac{\big\vert \sum_{i \in P} r_i \big\vert}{\sqrt{\sharp P}}.
\end{equation*}
\end{theorem}
\noindent
We note that other ``dyadic-type'' collections of scanning subsets (say, $p$-adic cubes, $p=2,3,\ldots$) will also lead to this asymptotic distribution, as can be seen from the proof in the appendix.

\paragraph{Local smoothing.}
If we want to allow the diffusivity to vary locally as for the inhomogeneous diffusion introduced in Section~\ref{sec_denois}, we have to modify the approach taken for the homogeneous diffusion above. It is the \emph{locality} of the multiresolution criterion that allows for such a modification: if the criterion is violated, we can not only conclude that the diffusivity was too high but from \eqref{mrcineq} we can also infer \emph{where}. Accordingly, each time a violation occurs for a multiresolution coefficient $\omega_P$ of some $P \in \mathcal{P}_n$, we reduce the local smoothing parameter $a(x)$ on that particular subset $P$ only, keeping the diffusivity at its current value in the rest of the image.

As for the homogeneous diffusion, we start by initializing $a(x)$ to a large constant. Then we compute the estimator and its residuals, check the multiresolution criterion and adapt the smoothing parameter locally. This strategy has first been introduced by \cite{Mei:2004}. Obviously, if the hypothesis is violated on a small subset $P$, we can \emph{logically} conclude that the hypothesis must also be violated on any superset $P' \supset P$. Note that this is not to say that if the \emph{empirical criterion} is violated on a subset that it will also be violated on all supersets -- this is clearly not true in general; however, if we decide against the hypothesis on such a subset we cannot accept the hypothesis on any superset. Hence we start by checking the multiresolution criterion on the smallest scale, considering larger sets only if no subset has shown a violation yet, thereby avoiding the diffusivity being reduced several times at the same spot. We then iterate the process until no further violations of the inequalities in \eqref{mrcineq} are detected, giving our final estimate.

Note that this results in the diffusivity being piecewise constant on the dyadic squares partitioning. In order to be able to adapt it to finer geometric features, we enhance the latter by adding so-called wedgelets. These are obtained by dividing a dyadic square into two parts by a straight line. We draw a certain set of lines through the image domain and add the resulting wedgelets to the partitioning. This is illustrated in Figure~\ref{figdyad} (right). For a detailed description of \emph{wedgelet partitionings} we refer to \citep{Do:1999} and \citep{Frie:2005}. Fast computation of the multiresolution coefficients of wedgelets can be done by a similar though somewhat more involved argument as in \eqref{rect}; for details we refer to \citep{FrieDeFueWi:2007}.

We use the wedgelet partitioning only to increase the flexibility when updating the diffusivity: if a violation is detected on a dyadic square $P$, all wedgelets subdividing that square are considered. If the criterion is fulfilled for all wedgelets contained in $P$ the smoothing parameter will be reduced on $P$. Otherwise, the smoothing parameter will be reduced on the wedgelet $W$ which yields the largest absolute value of the multiresolution coefficient $\omega_W$. Again, all supersets of $P$ will be ignored in that iteration. The final algorithm is described in Figure~\ref{algo}.

\begin{figure}[!tbp]
\framebox[\textwidth][c]{
\begin{minipage}{\textwidth}
\begin{algorithmic}
\newlength{\bull}\settowidth{\bull}{$\cdot$}
\newlength{\bif}\settowidth{\bif}{\textbf{if}}
\STATE $\cdot$ choose some $\delta > 0$, e.g. such that the error of the first kind is $5\%$\newline\hspace*{\bull} which can be obtained from simulations
\STATE $\cdot$ determine $\sigma$
\STATE $\cdot$ initialize $a:\{1, \dots, n\}^2 \rightarrow \real^+$ to a large constant
\LOOP
\STATE $\cdot$ obtain current reconstruction $\hat f(\cdot) = L_a y$, see \eqref{inhomdiffest}
\STATE $\cdot$ compute residuals $r_{ij} = y_{ij} - \hat f_{ij}$
\STATE $\cdot$ for each dyadic square $P$ determine its multiresolution coefficient
\begin{equation}
\omega_P = \frac{1}{\sqrt{\sharp\{ x_{ij} \in P\}}} \sum_{x_{ij} \in P} r_{ij}
\end{equation}
\IF{there is some $P$ on which the criterion is violated,\newline\hspace*{\bif} i.e. having $\vert \omega_P \vert > \sigma \sqrt{\delta \log n^2}$,}
\FORALL{such $P$}
\IF{there is some subset $P' \subset P$ with a violation}
\STATE $\cdot$ do nothing
\ELSE
\STATE $\cdot$ determine the multiresolution coefficients $\omega_W$ of all wedgelets\newline\hspace*{\bull} $W$ comprising this square
\IF{$\max_W \vert \omega_W \vert > \sigma \sqrt{\delta \log n^2}$, i.e. there is some wedgelet with a \newline\hspace*{\bif} large coefficient,}
\STATE $\cdot$ reduce $a$ on the wedgelet $W = \argmax_W \omega_W$
\ELSE
\STATE $\cdot$ reduce $a$ on $P$
\ENDIF
\ENDIF
\ENDFOR
\ELSE
\STATE $\cdot$ stop and return the current estimate $\hat f$, i.e. the first estimate\newline\hspace*{\bull}  without a violation
\ENDIF
\ENDLOOP
\end{algorithmic}
\end{minipage}
}
\caption{\label{algo}
Algorithm to determine local diffusivity and hence the reconstructed image.}
\end{figure}

For this procedure to be implemented, we need an estimator of the noise's standard deviation $\sigma$. Based on the normality assumption, and taking the small number of pixels at sharp edges into account, we use a robust estimator based on the median of absolute differences, namely
\begin{equation}
\label{sigmahat}
\hat\sigma = \frac{1}{2 \Phi^{-1}(0.75)} \median\big\{ \vert y_{i,j} - y_{i-1,j} - y_{i,j-1} + y_{i-1,j-1} \vert\ :\ i, j = 2, \dots, n \big\},
\end{equation}
where $\Phi$ denotes the cumulative distribution function of $\mathcal{N}(0, 1)$, cf. \citep{DaKo:2001}. For smooth signals, polynomially weighted, difference-based estimators might be more appropriate, see \citep{MuBiWaFrei:2005}; note that our estimator is indeed unbiased if $f$ is affine.

The algorithms described in this section clearly do not depend on the kind of penalty used, cf. \eqref{varinhomdiff}. All that is needed is a global or local smoothness parameter $a$, corresponding to the homogeneous or inhomogeneous diffusivity. In particular, we proceed in the same way for the TV regularization in \eqref{vartv}. See \citep{DaMei:2008} for an application of this approach to one- and two-dimensional weighted smoothing splines.

\section{Numerical details and results}
\label{sec_num}

In this section we show simulations where the noiseless image $f$ of \eqref{model} is the $256 \times 256$ pixel test image shown in Figure~\ref{figdata}(a), whose values spread over $[0,5]$. To this we added Gaussian noise according to \eqref{noise} with $\sigma = 1$, resulting in a signal-to-noise ratio of $5$, see Figure~\ref{figdata}(b).

In each iteration, the smoothing parameter was reduced where necessary through multiplication with a factor $\lambda < 1$ as described in Section~\ref{sec_adapt}. In order to speed up the algorithm we chose $\lambda$ according to the size of the multiresolution coefficient: the larger $\omega_P$, the smaller $\lambda$ was chosen, taking care not too quickly to reduce the diffusivity which would result in undersmoothing. Note that the inhomogeneous diffusion estimator can be found efficiently by solving \eqref{inhomdiffest} since the discretized Laplacian is given by a sparse band matrix, rendering Gauss-Seidel iterations an appropriate solution method.

The algorithms were implemented in Matlab in a modular fashion, thus allowing to use any localized regression method as discussed in Section~\ref{sec_denois}. Figure~\ref{figdiff} shows results for homogeneous and inhomogeneous diffusion. Clearly, in order to reconstruct fine details and sharp edges well, such that the multiresolution conditions \eqref{mrcineq} are satisfied, a very small diffusivity ($a = 0.934$) is needed when chosen globally; this results in considerable undersmoothing elsewhere in Figure~\ref{figdiff}(a). Choosing the diffusivity locally resolves this difficulty, allowing e.g. to strongly smooth large areas of the background or the circle in the upper left while not compromising small level details like the dots in the lower right or any sharp edges, see Figures~\ref{figdiff}(b) and~\ref{cut}(a). This behaviour of the local diffusivity as an edge-detector is very much apparent in Figure~\ref{figpar}(a).

\begin{figure}[!tb]
\begin{center}
\setlength{\tabcolsep}{2mm}
\begin{tabular}{ccl}
\includegraphics[width=150pt]{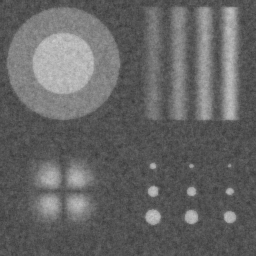}
&
\includegraphics[width=150pt]{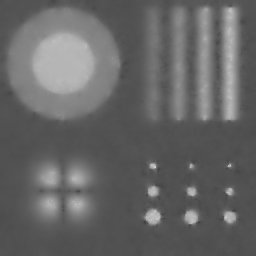}
&
\includegraphics[height=150pt]{sn5_pal.pdf}
\\
(a) & (b) &
\vspace*{-12pt}
\end{tabular}
\end{center}
\caption{\label{figdiff}Results of homogeneous (a) and inhomogeneous (b) diffusion.}
\end{figure}

For comparison, we also show results for the total-variation (TV) regularization introduced in \eqref{vartv}, again both for global and local smoothing parameters; see Figure~\ref{figtv}(a) and (b), respectively. The difficulty when choosing a global smoothing parameter is again well visible but the local smoothing parameter appears to have more difficulty adapting to the test object, permitting considerable undersmoothing on the larger dots in the lower right, cf. Figures~\ref{figpar}(b) and~\ref{cut}(b). The TV penalty was implemented by approximation with a differentiable functional, as described e.g. by \cite{Vo:2002}.

\begin{figure}[!tb]
\begin{center}
\setlength{\tabcolsep}{2mm}
\begin{tabular}{ccl}
\includegraphics[width=150pt]{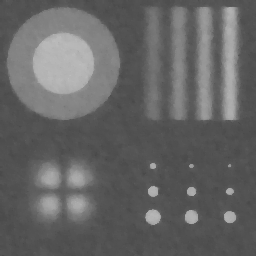}
&
\includegraphics[width=150pt]{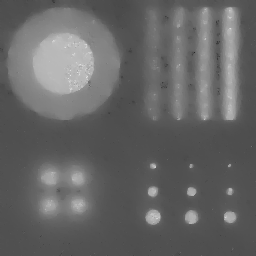}
&
\includegraphics[height=150pt]{sn5_pal.pdf}
\\
(a) & (b) &
\vspace*{-12pt}
\end{tabular}
\end{center}
\caption{\label{figtv}Results for TV regularization with global (a) and local (b) choice of the smoothing parameter.}
\end{figure}

\begin{figure}[!tb]
\begin{center}
\setlength{\tabcolsep}{2mm}
\begin{tabular}{cc}
\includegraphics[width=150pt]{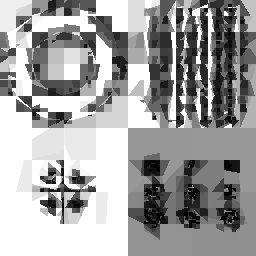}
\ \includegraphics[height=150pt]{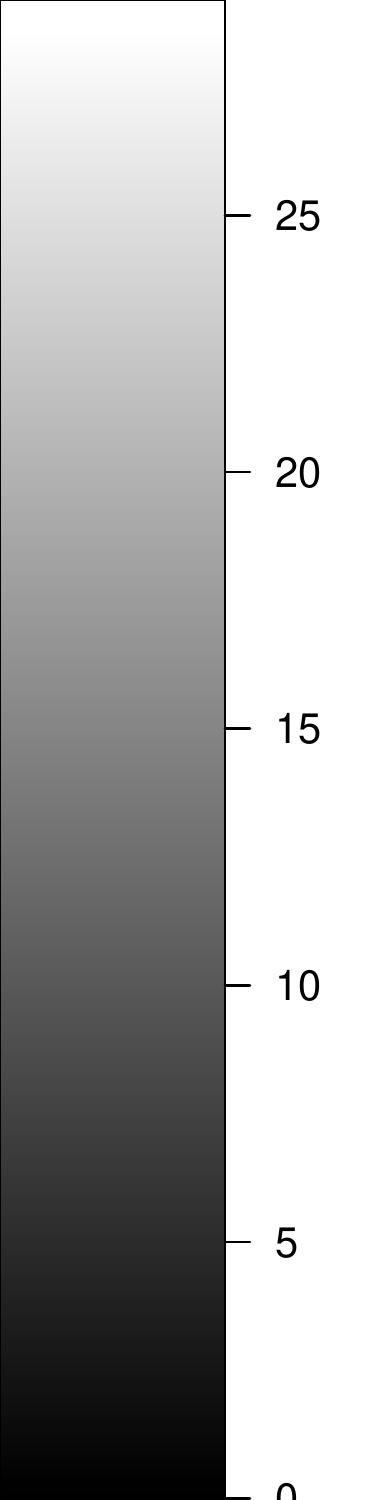}
&
\includegraphics[width=150pt]{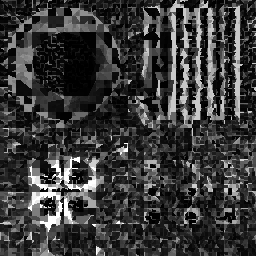}
\ \includegraphics[height=150pt]{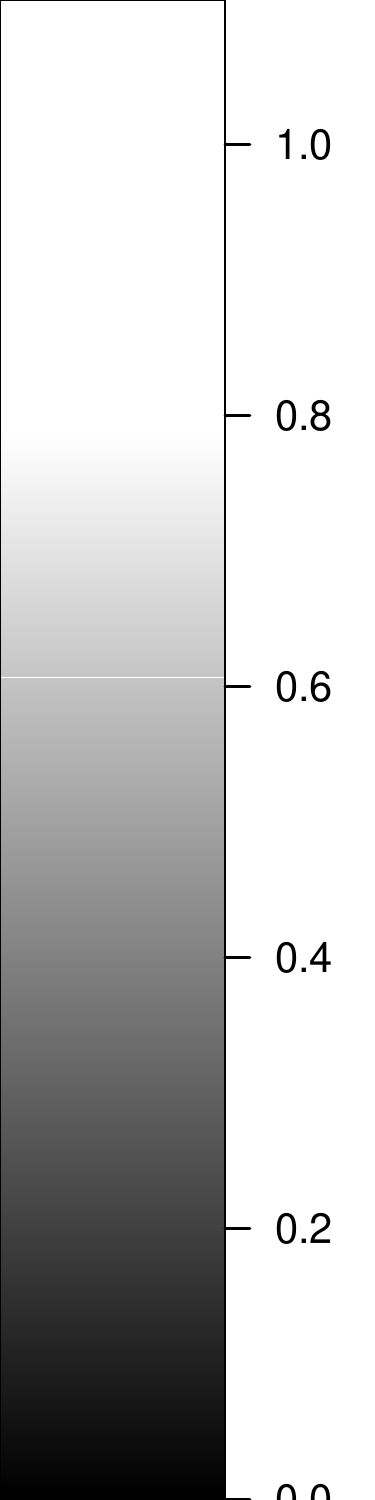}
\\
(a) & (b)
\vspace*{-12pt}
\end{tabular}
\end{center}
\caption{\label{figpar}Local smoothing parameters for diffusion (a) and TV penalty (b).}
\end{figure}

\begin{figure}[!tb]
\hspace*{-15mm}
\includegraphics[width=0.6\textwidth]{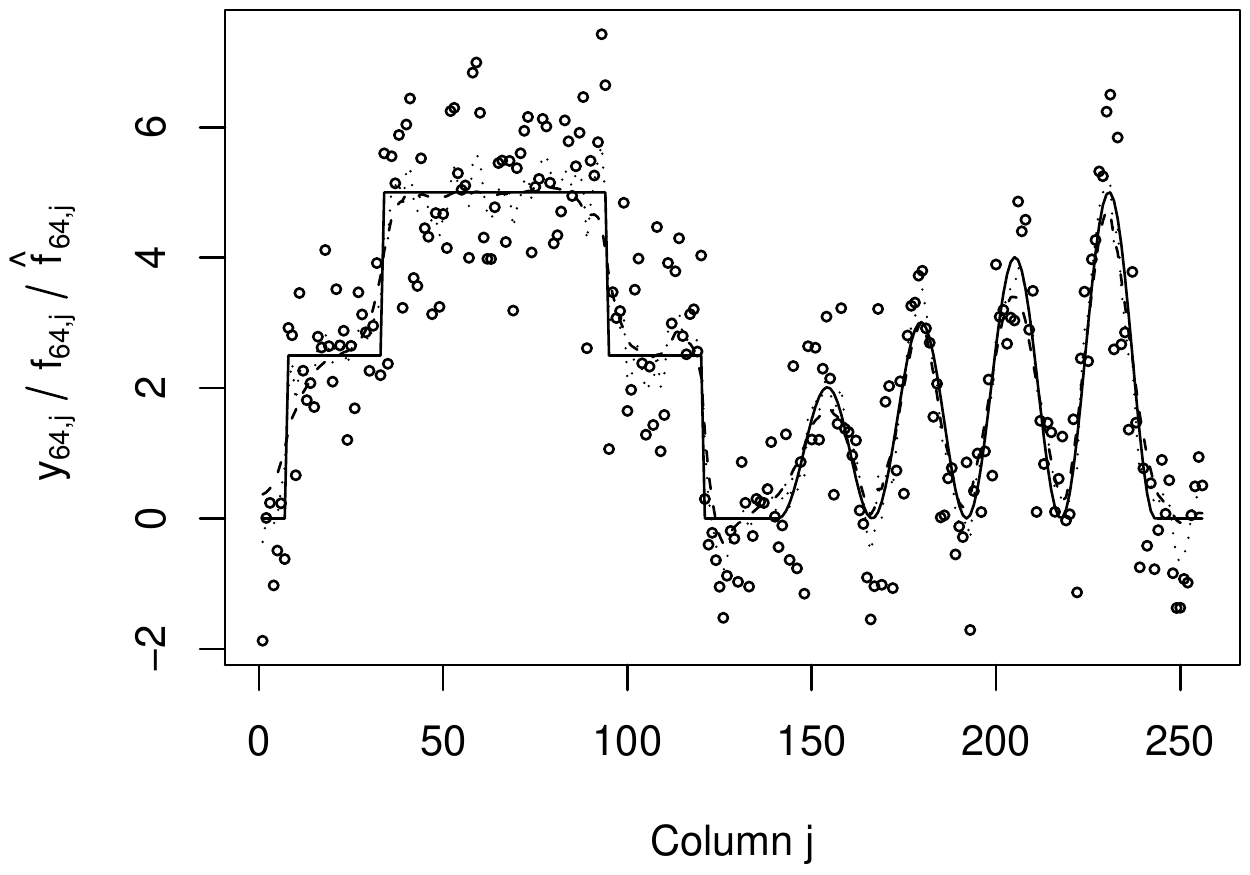}%
~~%
\includegraphics[width=0.6\textwidth]{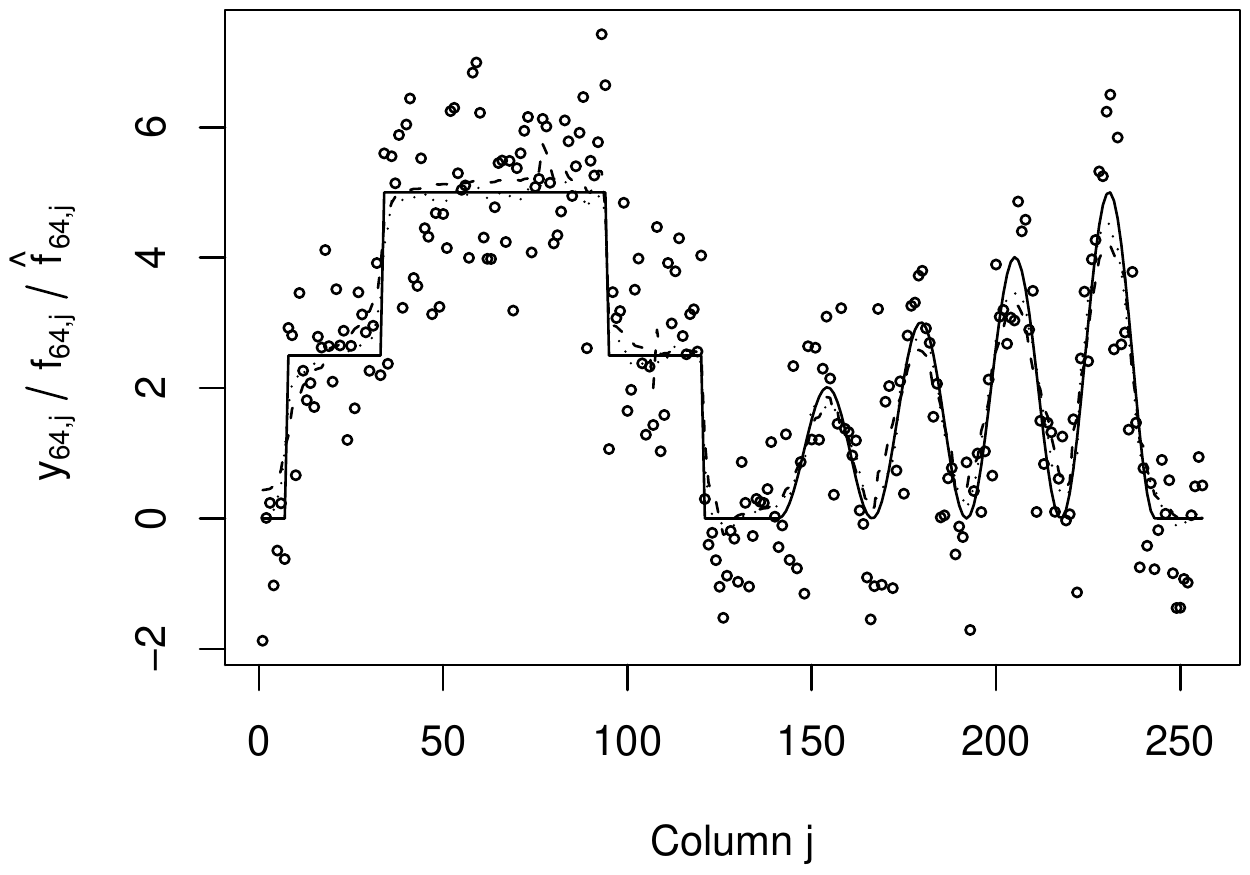}
\newline
\hspace*{26mm} (a) \hspace*{67mm} (b)

\caption{\label{cut}Cut along row 64 indicated in Figure~\ref{figdata}(a) for diffusion (a) and TV penalty (b); solid lines gives the true $f$, points the data $y$, dotted lines correspond to a global and dashed lines to a local choice of the smoothing parameter.}
\end{figure}

When the signal-to-noise ratio is reduced to 2, see the simulated data in Figure~\ref{noisier}(a), the reconstruction gets smoothed more strongly (b), since the residuals no longer carry enough statistically significant information to allow the oversmoothing of the smaller features to be detected, though still the only feature  not distinguishable from the noise is the smallest of the nine circles in the lower right.

\begin{figure}[!tb]
\begin{center}
\setlength{\tabcolsep}{2mm}
\begin{tabular}{ccl}
\includegraphics[width=150pt]{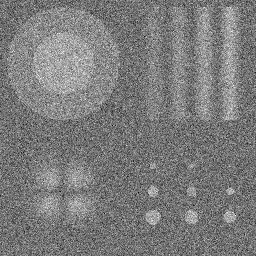}
&
\includegraphics[width=150pt]{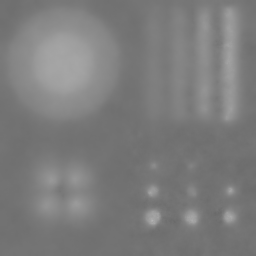}
&
\includegraphics[height=150pt]{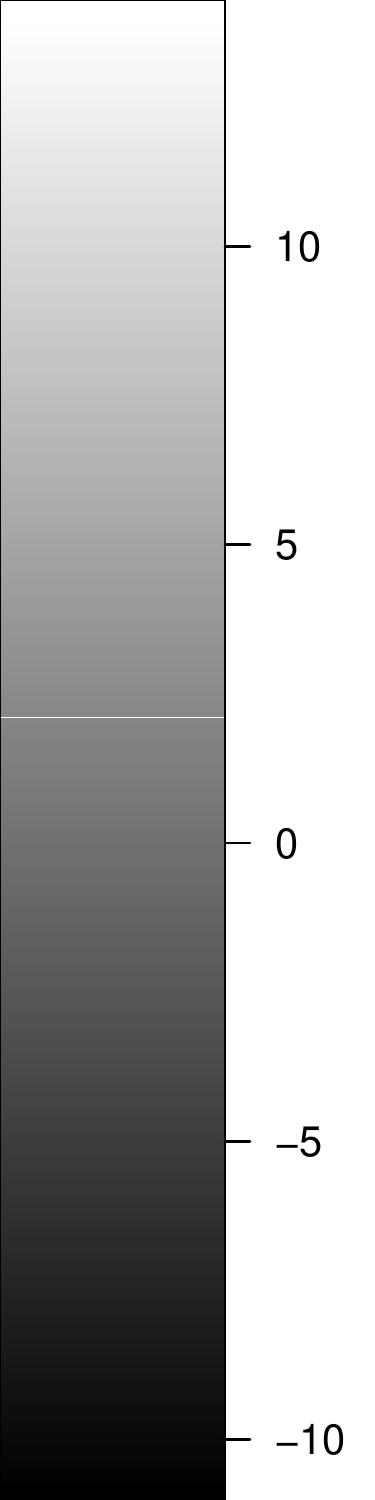}
\\
(a) & (b) &
\vspace*{-12pt}
\end{tabular}
\end{center}
\caption{\label{noisier}(a) Simulated data $y$ with noise level $\sigma = 2.5$. (b) Corresponding reconstruction using inhomogeneous diffusion.}
\end{figure}

Poisson data with high intensities can also be treated with this methodology as remarked at the end of Section~\ref{sec_mrc}; instead of determining a constant variance at the beginning of the algorithm, one uses the local variance predicted by the reconstruction to normalize the residuals. We also simulated this situation for intensities within $[50, 100]$, i.e. again with a signal-to-noise ratio of $5$, see Figure~\ref{poisson}. The reconstruction demonstrates the applicability of our approach also in this situation.

\begin{figure}[!tb]
\begin{center}
\setlength{\tabcolsep}{2mm}
\begin{tabular}{ccl}
\includegraphics[width=150pt]{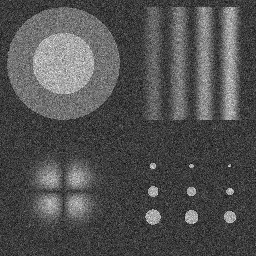}
&
\includegraphics[width=150pt]{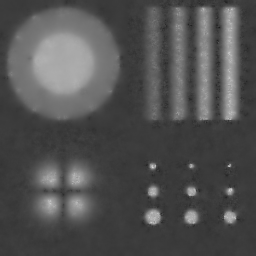}
&
\includegraphics[height=150pt]{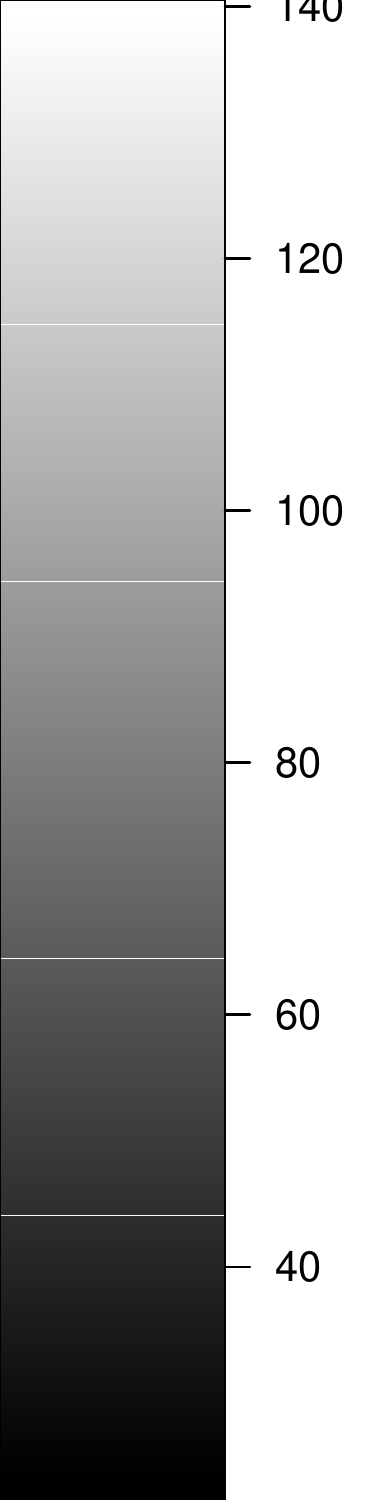}
\\
(a) & (b) &
\vspace*{-12pt}
\end{tabular}
\end{center}
\caption{\label{poisson}(a) Simulated Poisson data $y$ with intensities in $[50, 100]$. (b) Corresponding reconstruction using inhomogeneous diffusion.}
\end{figure}

We conclude this section with an application where the image has been obtained by a CCD camera attached to a confocal microscope, see Figure~\ref{quantum}, data courtesy of Emre Togan, Department of Applied Physics, Harvard University. The data set shows photoluminescence in a diamond sample, where high photoluminescence marks the so-called nitrogene-vacancy centres in the diamond. These are of major interest to quantum information science where they have been proposed as qubits, i.e. for storage, as they form a solid-state system whose spins can be manipulated at room temperature. This application therefore aims at removing noise from the image in order to aid the researcher in detecting these nitrogene-vacancy centres, such that subsequent experiments can be conducted on them, see \citep{DCJTMJZHL:2007} and the references therein. We note that our reconstruction reduces the noise considerably while keeping small-scale features, much to the satisfaction of the physicists involved.

\begin{figure}[!tb]
\begin{center}
\setlength{\tabcolsep}{2mm}
\begin{tabular}{ccl}
\includegraphics[width=150pt]{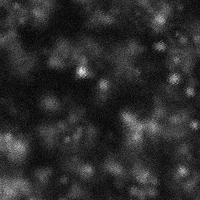}
&
\includegraphics[width=150pt]{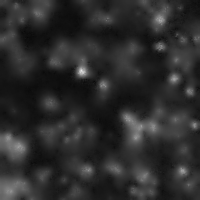}
&
\includegraphics[height=150pt]{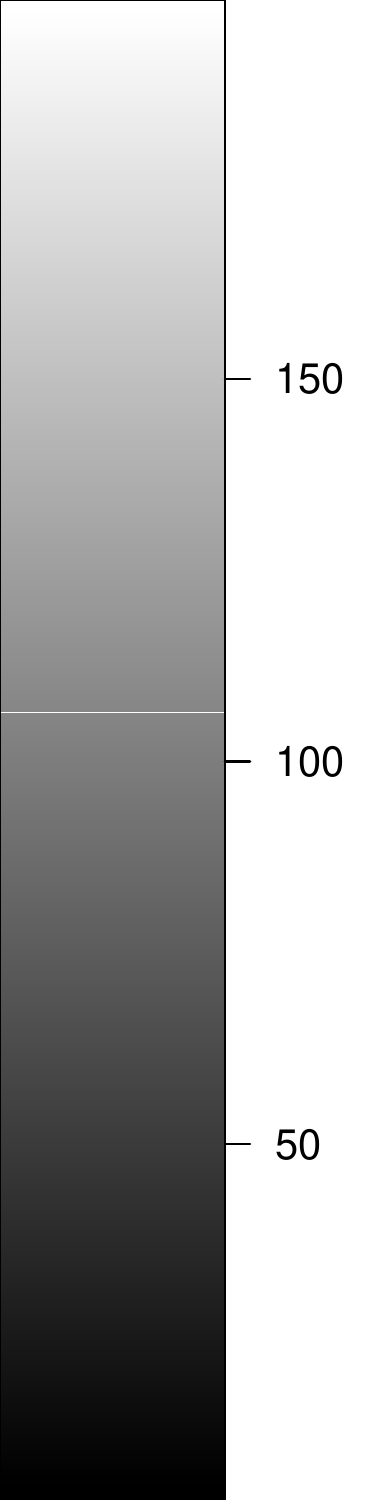}
\\
(a) & (b) &
\vspace*{-12pt}
\end{tabular}
\end{center}
\caption{\label{quantum}(a) Poisson data of photoluminescence. (b) Corresponding reconstruction using inhomogeneous diffusion.}
\end{figure}

\section{Discussion}
\label{sec_discuss}

We have demonstrated that the statistical multiresolution criterion allows to determine both global and local smoothing parameters in a fully data-driven procedure. For inhomogeneous diffusion, it acts as an edge-detector -- quite as expected: at edges, only very little smoothing is to be permitted for the criterion to be fulfilled there, while large regions that are well approximated by their first order Taylor series expansion may be smoothed strongly.

It appears that using a penalty that is more capable of dealing with sharp edges, like TV regularization, does not improve upon the inhomogeneous diffusion, or even performs worse. A possible explanation is that the multiresolution criterion itself is able to detect edges, thus eliminating the major drawback of the diffusion process -- whilst keeping its advantageous properties in smooth areas. TV regularization with a global parameter, however, already is capable of reconstructing sharp edges and can localize them, thus it cannot benefit as much from the multiresolution criterion's ability to choose the parameter locally, especially as the smoothness of the reconstruction cannot be improved.

We emphasize once again that the multiresolution analysis can be extended to other error distributions as well, see Section~\ref{sec_mrc}. Also, generalizations to higher dimensions are possible; in particular, three-dimensional images or movies could be treated efficiently, too.


\ifthenelse{\boolean{blind}}{}{
\section*{Acknowledgements}

T.~Hotz and P.~Marnitz gratefully acknowledge support by the German Federal Ministry of Education and Research, Grant 03MUPAH6. A.~Munk and Z.~Kabluchko acknowledge support by the German Research Foundation's FOR 916, and A.~Munk also by SFB 755.
}

\appendix

\section{Proofs}

The proof of Theorem~\ref{thm_gen} on page~\pageref{thm_gen} will be based on the following lemma which is known as Berman's Inequality, see e.g. \cite[Theorem 4.2.1]{LeLiRoo:1983}.
\begin{lemma}
Let $(\xi_1,\ldots,\xi_N)$  be a Gaussian vector with standard margins and covariance matrix $(\rho_{ij})_{i,j=1,\ldots,N}$, and let $\eta_1,\ldots,\eta_N$ be independent standard Gaussian variables. Then, for every $u>0$,
\begin{align*}
\P\left[\max_{1\leq i\leq N}\xi_i\leq u\right]-\P\left[\max_{1\leq i\leq N}\eta_i\leq u\right] \leq
&\frac 1 {2\pi} \sum_{1\leq i<j\leq N}\frac{|\rho_{ij}|}{\sqrt{1-\rho_{ij}^2}} \exp\left(-\frac{u^2}{1+\rho_{ij}}\right).
\end{align*}
\end{lemma}
\begin{proof}[\bfseries\upshape Proof of Theorem~\ref{thm_gen} on page~\pageref{thm_gen}.]
Let 
$u_N=a_N+b_N\tau$, with $a_N$ and $b_N$ as in \eqref{eq:defab}. We then have, using Berman's Inequality,
\begin{align*}
\lefteqn
{
\P\left[\max_{1\leq i\leq N}\xi_i^{(N)}\leq u_N\right]-\P\left[\max_{1\leq i\leq N}\eta_i\leq u_N\right]
}\\
&\leq  
O(1) \sum_{1\leq i<j\leq N} |\Cov(\xi_i^{(N)},\xi_j^{(N)})| \exp\left(-\frac{u_N^2}{1+\rho}\right)\\
&\leq 
O(N^{1+\eps}) \exp\left(-\frac{u_N^2}{1+\rho}\right)
\end{align*}
as  $N\to\infty$. Noting that $u_N\sim \sqrt{2\log N}$ and choosing $\eps$ small enough, we see that the right-hand side is $o(1)$ as $N\to\infty$.  Combining this estimate with~\eqref{eq:equatmain}, we obtain~\eqref{eq:stat_theo}, which finishes the proof of the theorem.
\end{proof}

\begin{proof}[\bfseries\upshape Proof of Theorem~\ref{thm_dyadic} on page~\pageref{thm_dyadic}]
We are going to apply Theorem~\ref{thm_gen} with $N=\sharp\mathcal{P}_n$ to the random vector $(\xi_P^{(N)})_{P\in \mathcal{P}_n}$, where
\begin{equation*}
\xi_P^{(N)}=\frac{\sum_{i \in P} r_i}{\sqrt{\sharp P}},\;\;\; P\in\mathcal{P}_n.
\end{equation*}

First we show that~\eqref{eq:bounded_cov} holds. Let $P_1$ and $P_2$ be two different dyadic subcubes of $K_n$. If $P_1\cap P_2=\emptyset$, then we have $\Cov(\xi_{P_1}^{(N)},\xi_{P_2}^{(N)})=0$, so that~\eqref{eq:bounded_cov} is fulfilled. If $P_1\cap P_2\neq \emptyset$, then we have either $P_1\subset P_2$ or $P_2\subset P_1$. Assuming that, say, $P_1\subset P_2$, we obtain (due to the dyadic structure) that $\sharp P_1 \leq \sharp P_2 / 2$, and hence,
\begin{equation*}
\Cov(\xi_{P_1}^{(N)},\xi_{P_2}^{(N)})
=\frac{\sharp (P_1\cap P_2)}{\sqrt{\sharp P_1 \sharp P_2}} =\frac{\sqrt{\sharp P_1}}{\sqrt{\sharp P_2}} \leq
\frac{1}{\sqrt{2}}.
\end{equation*}

Now we show that~\eqref{eq:sparse} holds. A dyadic cube with side length $2^k$ contains $2^{d(k-i)}$ dyadic cubes with side length $2^i$, where $0\leq i\leq k$. Therefore, the total number of dyadic cubes which are contained in a dyadic cube with side length $2^k$ is $\sum_{i=0}^{k-1}2^{d(k-i)}\leq 2^{d(k+1)}$. Further, the number of dyadic cubes in $\mathcal{P}_n$ having side length $2^k$ is $\lfloor n/2^k \rfloor^d$. Note also that, trivially, $\sharp\mathcal{P}_n\geq n^d$. So, we may estimate the cardinality on the left-hand side of~\eqref{eq:sparse} from above by
\begin{equation*}
2\sum_{k=0}^{\lfloor \log n/\log 2\rfloor} \lfloor n/2^k \rfloor^d \cdot 2^{d(k+1)}=O(n^d\log n)=O(\sharp\mathcal{P}_n^{1+\eps})
\end{equation*}
as $N\to\infty$ and for every $\eps>0$. This finishes the proof.
\end{proof}

\bibliographystyle{apalike}
\bibliography{mrc}

\end{document}